\documentclass[letterpaper, 12 pt]{article}  
\usepackage{amssymb}
\usepackage{amsmath,mathrsfs}
\usepackage{amsthm}
\usepackage{psfrag}
\usepackage{epsfig,fullpage}
\usepackage{color}

\title{
On Multisequences and their Extensions  
}

\author{Srinivasan Krishnaswamy*, H. K. Pillai}
\newtheorem{theorem}{Theorem}[section]
\newtheorem{corollary}[theorem]{Corollary}
\newtheorem{example}[theorem]{Example}
\newtheorem{remark}[theorem]{Remark}

\newtheorem{lemma}[theorem]{Lemma}
\newtheorem{definition}[theorem]{Definition}
\newtheorem{algorithm}[theorem]{Algorithm}
\newtheorem{problem}{Problem}

\newtheorem{subroutine}[theorem]{Subroutine}
\newcommand{\F}{\mathbb{F}}
\begin{document}
\maketitle

\begin{abstract}
In this paper we deal with the dimension of multisequences and related
properties. For a given multisequence $W$ and $R \in \mathbb{Z}_+$, we define
the $R-$extension of $W$. Further we count the number of multisequences $W$
whose $R-$extensions have maximum dimension and give an algorithm to derive
such multisequences. We then go on to use this theory to count the number of
Linear Feedback Shift Register(LFSR) configurations with multi input multi
output delay blocks for any given primitive characteristic polynomial and
also to design such LFSRs. Further, we use the result on multisequences to count the number of Hankel matrices of
any given dimension. 
\end{abstract}

\section{Introduction}
Linear recurring sequences, over a finite field $\F_q$, with maximum period have been shown to 
exhibit several important randomness properties such as $2$-level
autocorrelation property and span-$n$ property (all nonzero subsequences of
length $n$ occur once in every period)\cite{Golomb}. As a result they  
find applications in a wide array of areas including  cryptography 
\cite{Schneier}, error correcting
codes  \cite{peterson} and spread spectrum communication \cite{pickholtz}. 

An obvious extension of a sequence of scalars is a sequence of vectors over the
given finite field. Such a sequence of vectors is known as a multisequence.
Periodic multisequences and linear relations among the elements of such
sequences,  have been a subject of study for a considerable period of time
\cite{Daykin}, \cite{Mullen}, \cite{Yucas}. Generating multisequences with a
given minimal polynomial has been an important problem motivating papers like
\cite{Ecuyer}, \cite{Neider1} and \cite{Neider2}.  

In this paper, we start by deriving some basic theorems regarding
multisequences. We then introduce the concept of an $R-$extension of a
multisequence, for $R \in \mathbb{Z}_+^m$. We then derive an
algorithm to generate multisequences whose $R-$extensions have maximum dimension. Further
we derive a formula for the number of multisequences having this property.
As an application, we show that the problem of generating some special class of LFSR configurations for any given primitive characteristic polynomial is a special case of the above problem. We then go on to count the number of such LFSR configurations using the
formula derived for multisequences. Finally we demonstrate another application of the theory developed : a novel way to count the number of full rank Hankel matrices with entries from a given finite field. 

In the remainder of this paper, $\F_q$ denotes a field of cardinality $q$, where
$q$ is a prime power. $\F_q[s]$ denotes the ring of polynomials in $s$ with
coefficients from $\F_q$. The group of all full rank $n \times n$ matrices with
entries from $\F_q$ is denoted by $GL(n,\F_q)$. The cardinality of any set
$K$ is given by $|K|$. The set of positive integers is denoted by
$\mathbb{Z}_+$.
For some integer $i$, we denote the vector in
$\F_q^n$, with $1$ in the $i^{th}$ position and $0$ in the remaining positions,
by $e_i^n$. For any matrix $M$, we denote the submatrix of $M$, where
 the row indices run from $a$ to $b$ and the column indices run from $c$ to
 $d$, by $M(a:b,c:d)$. We denote the $a$-th row of a matrix by $M(a,:)$. A matrix with row vectors $v_1,v_2,
\ldots, v_n$ is represented as $[v_1;v_2;\ldots;v_n]$. The column span of a
matrix $M$ is
denoted by $colspan(M)$. 
\section{Multisequences }
\label{Multisequences}
We define a sequence $S$ in $\F_q$ as map from $\mathbb{Z}$ to $\F_q$. 
A sequence $S =\{S(k)\}_{k \in \mathbb{Z}}$, in a finite field $\F_q$
is called periodic if there exists an integer $r$ such that $S(k+r) = S(k)$ for
all $k$. 
The smallest such nonnegative $r$ is called the period of the sequence. There
are linear relations amongst the elements of a periodic sequence. One obvious
example of such a relation being $S(k+r) = S(k)$. A general form of such a
relation is
{\small
\begin{equation}
\label{eq1}
S(k+n) = a_{n-1}S(k+n-1) + a_{n-2}S(k+n-2) + \cdots + a_0S(k) ~~ \forall k
\textrm{ where
}
a_i \in \F_q.
\end{equation}
}
These are called  Linear Recurring Relations (LRRs). The integer $n$ in equation
\eqref{eq1} is called the
order of the LRR. Given an LRR we can uniquely associate a monic polynomial
with it. For example the polynomial associated with the LRR in equation
\eqref{eq1} is $p(s) =  s^n - a_{n-1}s^{n-1} - a_{n-2}s^{n-2} - \cdots - a_0 $.
Since we are dealing with periodic sequences, without loss of generality, we can
assume that $a_0 \neq 0$  \cite[Theorem 6.11]{lidl}.

It is easy to check that all polynomials associated with LRRs of a given
sequence $S$, form an ideal $\mathfrak{I}_S$ in the polynomial ring $\F_q[s]$.
Since $\F_q[s]$ is a principal ideal domain, every ideal has a unique monic
generating polynomial. The generating polynomial of $\mathfrak{I}_S$ is called
the minimal polynomial of the sequence $S$.
The degree of the minimal polynomial is called the linear complexity of the
sequence. 

Given an LRR of degree $n$, there are many sequences that satisfy this relation.
In fact, the collection of all sequences that satisfy this relation form a
vector space over $\F_q$. The maximum possible period of sequences in this
vector space is equal to the order of the polynomial associated with the LRR. In
particular, if the polynomial associated with 
the LRR is a primitive polynomial of degree $n$, then every nonzero sequence in
the corresponding vector space has a period equal to $q^n - 1$ ( \cite[Theorem
6.33]{lidl}).

Consider a sequence of linear complexity $n$. Given $n$ consecutive elements of
the sequence, every subsequent element can be generated using the LRR
corresponding to the minimal polynomial. The vector consisting of $n$
consecutive elements of the sequence is called the state vector of the sequence.
We denote the $i$-th state vector of the sequence by $x(i)$ i.e., $x(i) =
[S(i),S(i+1),\ldots,S(i+n-1)]$. Observe that if the minimal polynomial of the
sequence is primitive then the sequence has $q^n-1$ different state
vectors i.e., every nonzero vector in $\F_q^n$ is a state vector of the
sequence. 

Let $\sigma S$ denote the sequence got by shifting the sequence $S$ once
to the left i.e., $\sigma S(k) = S(k+1)$. The $k$-th state vector of
$\sigma S$ is denoted by $\sigma x(k)$. Therefore $\sigma x(k) = x(k+1)$.
Observe that  $\sigma x(k) = x(k+1) = x(k)A $, where
\begin{eqnarray*}
\label{A}
  A = \left[
 \begin{matrix}
 0 & 0 & \ldots & 0 & a_0\\
 1 & 0 & \ldots & 0 & a_1\\
 0 & 1 & \ldots & 0 & a_2\\
 \vdots & \vdots & \ddots & \vdots & \vdots\\
 0 & 0 & \ldots & 1 & a_{n-1} 
 \end{matrix}
 \right]
   \in \F_q^{n \times n}
\end{eqnarray*}
This matrix is the companion matrix of the polynomial $p(s)= s^n -(\sum_0^{n-1}
a_is^i)$.  Observe that the 
companion matrix associated to the polynomial is unique. The sequence obtained by
shifting $S$ 
$\ell$ times to the left, is denoted by $\sigma^\ell S$ i.e $\sigma^\ell
S(k) = S(k+\ell)$.

Similar to sequences, we define a multisequence in $\F_q^m$ as a map from
$\mathbb{Z}$ to $\F_q^m$. 
 A multisequence $W = \{W(k)\}_{k \in \mathbb{Z}}$   is called
periodic if exists a finite
integer $r$ such that $W(k+r)   = W({k})$, for all $k$. As in the case
of scalar sequences, there exist linear  recurring relations between the
elements of the
multisequence. These relations are of the form
{\small
\begin{equation}
\label{veceq1}
W(k+n) = a_{n-1}W(k+n-1) + a_{n-2}W({k+n-2}) + \cdots + a_0W(k)\,\, \forall
k\textrm{ where
}
a_i \in \F_q
\end{equation}
}
 Analogous to scalar sequences, one can define
a polynomial $p(s) =  s^n - a_{n-1}s^{n-1} - a_{n-2}s^{n-2} - \cdots - a_0 $,
which can be associated with every LRR.
Again, the polynomials associated to all LRRs of a 
given periodic multisequence, form an ideal in the principal ideal domain
$F_q[s]$ and the monic
generator of this ideal is called the minimal polynomial of the multisequence.
The degree of the minimal polynomial is defined as the linear complexity of the
multisequence.    
      
  The $i^{\textrm{th}}$ component of each vector in $W$ gives a sequence of
scalars in $\F_q$. We call this sequence the
$i^{\textrm{th}}$ component sequence of $W$, denoted by $W_i$. Clearly, the
minimal polynomial of the multisequence is the least common multiple of the
minimal polynomials of the component sequences. Therefore, the minimal
polynomials of each of the component sequences divide the minimal polynomial of
the multisequence. Hence, if the minimal polynomial of the multisequence is an
irreducible polynomial $p(s)$, each of the nonzero component sequences also have
$p(s)$ as
their minimal polynomial.

Note that a multisequence, with linear complexity $n$ is completely determined
by the first $n$ terms (vectors).  The state of a multisequence can 
therefore be thought of as $n$ consecutive elements (vectors) of the 
multisequence. Each state is thus an $m\times n$ matrix.
 Thus, we have a sequence of  matrix states associated with every
multisequence. We denote the $k$-th matrix state of the multisequence $W$ by
$M_W(k)$, i.e. $M_W(k) = [W(k),W({k+1}),\ldots,W({n+k-1})]$. The
$i^{\textrm{th}}$ row, denoted by $x_i(k)$, of $M_W(k)$ is the $k-$th state
vector of the component sequence $W_i$. 
In the following theorem
we prove that for a periodic multisequence $W$, all matrix states have
the same column span.

\begin{lemma}
\label{span}
For a periodic multisequence, the column span of the matrix states is an
invariant.
\end{lemma}

\begin{proof}
Consider a periodic multisequence $W$. It is enough to show that
$colspan(M_W(k)) = colspan(M_W(k+1))$, for any given integer $k$.
 Let the minimal polynomial of the multisequence be $p(s) =  s^n -
a_{n-1}s^{n-1} - a_{n-2}s^{n-2} - \cdots - a_0 $. Since $W({k+n}) = a_0W(k)
+ a_1W({k+1}) + \ldots + a_{n-1}W({k+n-1})$, therefore $W({k+n}) \in
colspan(M_W(k))$.
Thus,
\begin{eqnarray*}
 colspan(M_W(k+1)) \subseteq colspan(M_W(k)) 
\end{eqnarray*}
Since $a_0 \neq 0$, $W(k) = \frac{1}{a_0} (W({k+n}) - a_1 W({k+1}) - a_2W({k+2})
- \ldots -
a_{n-1}W({k+n-1}))$, i.e., $W(k) \in colspan(M_W(k+1))$. Hence 
\begin{eqnarray*}
 colspan(M_W(k)) \subseteq colspan(M_W(k+1)) 
\end{eqnarray*}
Therefore $colspan(M_W(k+1)) = colspan(M_W(k))$. Hence proved.  

%
\end{proof}

We define the dimension of a multisequence as follows
\begin{definition}
The dimension of a multisequence $W$ is defined as the rank of its matrix
states.
\end{definition}

As in the case of scalar sequences, any nonzero multisequence with a primitive
minimal polynomial $p(s)$ of degree $n$, has a period of $q^n-1$.
In this paper, we henceforth assume that the multisequences considered
have primitive minimal polynomials.

The first problem we address is the following:
\begin{problem}
 Given a positive integer $\ell$ and a primitive polynomial $p(s)$ of degree
$n$, how many
multisequences of dimension $\ell$ exist in $ \F_q^m$ with $p(s)$ as its minimal
polynomial. (Clearly $0 \leq \ell \leq min(m,n)$)
\end{problem}
 
 {\bf Two multisequences are considered the same if they
are
shifted versions of one another}, i.e., the multisequence $W$ is the same as 
its shifted version $\sigma^r W$  for any $r \in \mathbb{Z}$. 

 We denote the collection of $\ell$ dimensional
subspaces of
$\F_q^m$ by $G(\ell,m,\F_q)$. The cardinality of  $G(\ell,m,\F_q)$ is given by
\begin{equation}
 |G(\ell,m,\F_q)| =
\frac{(q^m-1)(q^m-q)\ldots(q^m-q^{\ell-1})}{(q^\ell-1)(q^\ell-q)\ldots(q^\ell-q^
{\ell-1})}
\end{equation}
 
\begin{lemma}
\label{NoMult}
Given a primitive polynomial $p(s)$ of degree $n$, the number of
multisequences in $\F_q^m$, with minimal polynomial $p(s)$, having dimension
$\ell$ is $ |G(\ell,m,\F_q)|
\times (q^n-q)(q^n-q^2)\ldots(q^n-q^{\ell-1}) $.
\end{lemma}

\begin{proof}
 Given a multisequence $W$ of dimension $\ell$, by Lemma \ref{span},  
the column space of the matrix states $M_W(k)$ is a unique $\ell$
dimensional
subspace of $\F_q^m$.
Observe that there are $|G(\ell,m,\F_q)|$ subspaces of $\F_q^m$ that have
dimension $\ell$. Consider any one such $\ell-$ dimensional space $V$. Fix  a
basis
for $V$, say $v_1,v_2,\ldots, v_\ell$, where $v_i \in \F_q^m$. Let $T$ be
the matrix $ T = [v_1,v_2,\ldots,v_\ell]$. Any $M \in
\F_q^{m \times n}$ whose column span is $V$ can then be written as $M=TB$,
where $B \in \F_q^{\ell \times n}$. The number of such matrices $B$ is equal to 
$(q^n-1)(q^n-q)\ldots(q^n-q^{\ell-1})$ (choosing $\ell$ independent vectors in
$\F_q^n$). As the polynomial $p(s)$ is primitive, 
each multisequence has $q^n - 1$ distinct matrix states. Thus,
the number of multisequences with column span $V$ is
equal to $\frac{(q^n-1)(q^n-q)\ldots(q^n-q^{\ell-1})}{q^n-1} =
(q^n-q)(q^n-q^2)\ldots(q^n-q^{\ell-1})$. Therefore, given a primitive
polynomial $p(s)$ of degree $n$, the number of
multisequences in $\F_q^m$ with minimal polynomial $p(s)$ having dimension
$\ell$ is $ |G(\ell,m,\F_q)|
\times (q^n-q)(q^n-q^2)\ldots(q^n-q^{\ell-1}) $.
\end{proof}

If a multisequence in $\F_q^m$ has dimension $m$, its component sequences are
linearly independent. We can therefore give the following corollary to Lemma
\ref{NoMult}.

\begin{corollary}
\label{maincor}
Given a primitive minimal polynomial $p(s)$ of degree $n$, the number of
multisequences in $\F_q^m$, with minimal polynomial $p(s)$, having linearly
independent component sequences is
$(q^n-q)(q^n-q^2)\ldots(q^n-q^{m-1})$.
\end{corollary}

\section{Extensions of Multisequences}

We now look to extend an $m$-dimensional 
multisequence $W$ in $\F_q^m$  to an $r$-dimensional multisequence in $\F_q^r$ 
where $r > m$. Further, we impose a condition that the minimal polynomial of the new multisequence is 
the same as the minimal polynomial of $W$. An obvious way of keeping the 
minimal polynomial unchanged is by appending to $W$ its component sequences 
or their linear combinations. Thus $W_j = \sum_{i=1}^m a_iW_i$ for $j > m$, 
where $a_i \in \F_q$. The extended multisequence however continues to have dimension $m$.
On other hand, appending $W$ with shifted versions of the component sequences
may perhaps increase the dimension of the multisequence.

Let $R = (r_1, ... , r_m) \in \mathbb{Z}_+^m$, with $\sum r_{k} = r$. We define
the $R$-extension of the multisequence $W$ in $\F_q^m$ as the multisequence
$W_R$ in $\F_q^r$, whose component sequences are obtained from the component 
sequences of $W$ in the following order : $W_1, \sigma W_1, ... \sigma^{r_1 - 
1}W_1, W_2, \sigma W_2, ... \sigma^{r_2 - 1}W_2$ $, ... W_i, \sigma W_i, ... 
\sigma^{r_i - 1}W_i, ... W_m, \sigma W_m, ... \sigma^{r_m - 1}W_m$. Clearly, the
minimal polynomial of the 
multisequences $W_R$ and $W$ are the same. We can
therefore ask the following question.

\begin{problem}
\label{mainProblem}
 Given $R = (r_1,r_2,\ldots , r_m) \in \mathbb{Z}_+^m$, 
with $\sum r_k = r$, how many multisequences $W$ of rank $m$ in $\F_q^m$ give 
$R$-extended multisequences in $\F_q^r$ having dimension $r$?
\end{problem}

The solution to this problem is given by the following theorem.

 \begin{theorem}
\label{maintheorem}
Let $R = (r_1,r_2,\ldots,r_m) \in \mathbb{Z}_+^{m}$ such that 
$r = \sum r_i$ and let $p(s)$ be a primitive polynomial of degree
$n$. The number  of multisequences in
$\F_q^m$ with minimal polynomial $p(s)$ whose $R-$extensions have dimension $r$
is equal to $(q^n-q^{r-m+1})(q^n - q^{r-m+2})\ldots(q^n - q^{r-1}) $.
\end{theorem}

In the remainder of this section we give a constructive proof to this theorem.
Starting with a multisequence in $\F_q^m$ with dimension $m$, we recursively
generate a series
of $r-m$ multisequences in $\F_q^m$ culminating in a desired multisequence
whose $R-$extension has dimension $r$. We first prove a few preparatory
results which when put together gives us the constructive proof.   
 
 For any $G = (g_1,g_2,\ldots,g_m) \in \mathbb{Z}_+^m$, let $G_{max}
 = max_ig_i$. 
Let $\Phi$ define the following map from $\mathbb{Z}_+^m$ to
$\mathbb{Z}_+^m$. 
\begin{eqnarray*}
 \Phi(g_1,g_2,\ldots,g_m) = (g_1,g_2,\ldots,g_{c-1},g_c -
1,g_{c+1},\ldots,g_m)\\ \textrm{ where } c \textrm{ is the smallest integer such
that }g_c =G_{max}
\end{eqnarray*}
Note that the repeated action of $\Phi$ on any element of
$\mathbb{Z}_+^m$ eventually gives ${\bf 1} = (1,1,\ldots,1)$. Thus, given
$R = (r_1,r_2,\ldots,r_m) \in \mathbb{Z}_+^m$, $\Phi$
defines a unique path from $R $ to ${\bf 1}$. We call this path the `$R-$road'.
\begin{example}
 The $R-$road for $R = (3, 2, 5, 4, 1)$ 
is $(3,2,4,4,1)$ $(3,2,3,4,1)$ $(3,2,3,3,1)$ $(2,2,3,3,1)$ $(2,2,2,3,1)$ $
(2,2,2,2,1)$ $(1,2,2,2,1)$ $(1,1,2,2,1)$ $(1,1,1,2,1)$ $(1,1,1,1,1)$.
\end{example}
Clearly given any point $G=(g_1,g_2,\ldots,g_m)$ on an $R-$road, for any
other point
$Q = (q_1,q_2,\ldots,q_m)$ lying on the path from $R$ to $G$, $q_i
\geq g_i$ $\forall i$. Besides, the map $\Phi$ ensures
the following: \begin{itemize}{\item[-] If $i<j$, $g_i> g_j$ if and only if $g_i
> r_j$}\end{itemize}

We now look to retrace the $R$-road from ${\bf 1}$ to $R$. As a first step, we
prove the following lemma.

\begin{lemma}
For every point $G = (g_1,g_2,\ldots,g_m) \neq R$ on the $R-$road, there exists
a coordinate $g_c$ which satisfies
at least one of the following conditions:
\begin{enumerate}
 \item $g_c = G_{max}-1$ and $g_c<r_c$.
 \item $g_c = G_{max}$ and $g_c<r_c$.
\end{enumerate}
\end{lemma}

\begin{proof}
 For every point $G = (g_1,g_2,\ldots,g_m) \neq R$ on the $R-$road, there exists
a unique point $G^*$ on the $R-$road such that $\Phi(G^*) = G$.
Now, $G^* = (g_1,g_2,\ldots,g_{c-1},g_c+1,g_{c+1},\ldots,g_m)$, where $g_c+1
\geq g_i$ $\forall i \neq c$. Also, since $G^*$ is on the path from $R$ to
${\bf 1}$, $g_{c}+1 \leq r_c$. Therefore, $g_c<r_c$. If $g_c+1>g_i$
$\forall i \neq c$ then $g_c = G_{max}$. If instead, there exists an $i$ such
that $g_c+1 = g_i$, then $g_c = G_{max}-1$. Hence proved.
\end{proof}

We therefore have the following definition.
\begin{definition}
Consider an $R = (r_1,r_2,\ldots,r_m) \in \mathbb{Z}_+^m$. For every point $G =
(g_1,g_2,\ldots,g_m) \neq R$, 
on the $R-$road the active coordinate is defined as follows: 
\begin{enumerate}
\item
If there exists a coordinate $g_c$ such that $g_c = G_{max}-1$ and $g_c <r_c$,
then the active coordinate is the coordinate corresponding to the largest such
$c$. 
\item
In the event of there being no $g_c$ that satisfies point 1, the active
coordinate is the
coordinate corresponding to the largest $c$ such that $g_c=G_{max}$ and
$g_c<r_c$.
\end{enumerate}
\end{definition}

It can be easily seen that one can traverse the $R-$road backwards from ${\bf 1}$ to $R$
by repeatedly
incrementing the active coordinate at every point. This is demonstrated in the
following example:
\begin{example}
 Let $R = (3, 2, 5, 4, 1)$. Starting from ${\bf 1}$ the $R-$road is traversed
backwards as follows: (At every point the active coordinate is underlined) 
$(1,1,1,{\underline1}, 1)$ $(1,1,{\underline 1},2,1)$ $(1,{\underline1},2,2,1)$
$({\underline1},2,2,2,1)$ $(2,2,2,{\underline2},1)$ $(2,2,{\underline2},3,1)$
$({\underline2},2,3,3,1)$ $(3,2,3,{\underline3},1)$ $(3,2,{\underline3},4,1)$
$(3,2,{\underline4},4,1)$ $(3, 2, 5, 4, 1)$.
\end{example}
 
Detecting the active
coordinate of any point $G$ involves the following steps:
\begin{itemize}
 \item Find $G_{max}$.
 \item Find the largest $i$ such that the $i-$th coordinate has value
$G_{max}-1$ and is less than $r_i$.
 \item In the event of there being no $i$ satisfying the preceding condition, find
the largest $j$ such that
the $j-$th coordinate has value $G_{max}$ and is less than $r_j$.
\end{itemize}

Notice that each of the above steps can be implemented in $O(m)$ operations.

For generating multisequences with $R$-extensions having maximum dimension, 
we travel backwards along the $R$-road from ${\bf 1}$ to $R$. During this
backward 
traversal, at every point $G$ on the $R$-road, we recursively generate a 
multisequence whose $G$-extension has maximum dimension.

We now make the following observation: Given a matrix $A\in\F_q^{\ell \times
\ell}$ in the companion form and a vector
$x = (b_1, b_2, \ldots, b_\ell) \in \F_q^{\ell}$, for $k < \ell$, $xA^k$ has the
following form
\begin{eqnarray*}
 xA^k = (b_{k+1},b_{k+2},\ldots,b_\ell,\underbrace{*,*,\ldots,*}_{k\textrm{
entries}})
\end{eqnarray*}
 where the $*$s are elements in $\F_q$, whose values depend on the matrix $A$.
Therefore, the matrix $[x; xA; \ldots; xA^{k-1}]$ has the following
structure.
\begin{eqnarray*}
\left[
 \begin{array}{cccccccc}
  b_1 & b_2 & \ldots & b_{\ell-k+1} &   b_{\ell-k+2} & \ldots & b_{\ell-1} &
b_{\ell}\\
  b_2 & b_3 & \ldots & b_{\ell-k+2} &   b_{\ell-k+3} & \ldots & b_{\ell} & *\\
  \vdots& \vdots &\vdots &\vdots &\vdots &\vdots &\vdots &\vdots\\
  b_k & b_{k+1} & \ldots & b_{\ell} & * & \ldots & * & *\\
 \end{array}
\right]
\end{eqnarray*}    

For any $G \in \mathbb{Z}_+^m$, let $N(G,k)$ denote the number of multisequences
in $\F_q^m$ with a given primitive minimal polynomial of degree $k$, whose $G-$extensions have maximum dimension.
\begin{theorem}
\label{starLemma}
Let $R = (r_1,r_2,\ldots,r_m)$, and let $G = (g_1,g_2,\ldots,g_m)$ and $\Phi(G)$
be consecutive points on the $R-$road. Then,  
\begin{equation}
 N(G,{k}) = q^{m-1}N(\Phi(G),{k-1})
\end{equation}
where $k$ is any integer greater than $g = \sum_{i=1}^mg_m$.
\end{theorem}

\begin{proof}
Let $c$ be the smallest integer such that $g_c =
G_{max}$. Therefore $\Phi(G) =(g_1, g_2, \ldots, g_{c-1}, g_c-1, g_{c+1},
\ldots, g_m)$. Let $W$ be a multisequence in $\F_q^m$ whose minimal 
polynomial $p_{k-1}(s)$ is a primitive polynomial of degree $k-1$. Further 
assume that the $\Phi(G)-$extension of $W$ has dimension $g-1$.
 Each matrix state of $W$ is therefore a matrix in $\F_q^{m\times (k-1)}$ with
full row rank. As $p_{k-1}(s)$ is a primitive polynomial 
of degree $k-1$, there is a matrix state $M$ of $W$, whose $c-$th row is 
$e_{k-1}^{k-1} = (0, 0, ..., 0, 1)$. For $i \neq c$, let
$x_i=[b_{i1},b_{i2},\ldots,b_{1(k-1)}]$ be
the $i-$th
row of this $M$. Therefore, $M = [x_1; x_2; \ldots; x_{c-1}; e_{k-1}^{k-1};
x_{c+1};
\ldots;x_m]$. Now expand $M$ to a matrix $M^* \in \F_q^{m \times
k}$ as follows:
\begin{enumerate}
 \item For every $i \neq c$, append the $i-$th row of $M$ with any element $d_i$
of $\F_q$. Therefore, the $i-$th row of $M^*$ is $x_i^* = (x_i,d_i) \in
\F_q^{k}$, for some $d_i \in \F_q$.
\item Let the $c-$th row of $M^*$ be $e_{k}^{k}$ i.e., $(0,0,\ldots,0,1)$.
 
\end{enumerate}
Let $p_k(s)$ be a primitive polynomial of degree $k$. Using $M^*$ 
as a matrix state, one can generate a multisequence $W^*$ whose minimal 
polynomial is $p_k(s)$. We claim that $W^*$ has a $G-$extension with dimension 
$g$.  

As $ M$ is a matrix state of $W$, the following matrix $M_{\Phi(G)}$ is a
 matrix state of the $\Phi(G)-$ extension of $W$:
\begin{eqnarray*}
 M_{\Phi(G)} &=& [\textcolor{green}{x_1;x_1A_{k-1};\ldots;x_1A_{k-1}^{g_1-1}};
\textcolor{red}{x_2;x_2A_{k-1};\ldots;x_2A_{k-1}^{g_2-1}}; \ldots;
 \textcolor{blue}{x_{c-1};x_{c-1}A_{k-1};\ldots;x_{c-1}A_{k-1}^{g_{(c-1)}-1}}
;\\&&
e_{k-1}^{k-1};e_{k-1}^{k-1}A_{k-1};\ldots;e_{k-1}^{k-1}A_{k-1}^{g_c-2}
;\textcolor {cyan} {
x_{c+1};x_{c+1}A_{k-1};\ldots;x_{c+1}A_{k-1}^{g_{(c+1)}-1}};
\ldots;\textcolor{magenta}{x_m;x_mA_{k-1};\ldots;}\\&&
\textcolor{magenta}{x_mA_{k-1}^{g_m-1}}]\\
&& \textrm{where } A_{k-1} \textrm{ is the companion matrix of the polynomial
}p_{k-1}(s)
\end{eqnarray*}
 The $c-$th block of rows of
$M_{\Phi(G)}$ has the following structure:
 \begin{eqnarray*}
\left[ \begin{array}{cccc|cccc}
  0 & 0 & \cdots & 0 & 0 & \cdots & 0 & 1\\
  0 & 0 & \cdots & 0 & 0 & \cdots & 1 & *\\
  \vdots & \vdots & \vdots & \cdots & \vdots & \vdots & \vdots & \vdots\\
  0 & 0 & \cdots &  0& 1 & \cdots & * & *\\
 \end{array}\right]\in \F_q^{(g_c-1) \times (k-1)}
\end{eqnarray*}

For $1\leq i\neq c \leq m$, let $x_i = (b_{i1},b_{i2}, \ldots,b_{i(k-1)})$.
The corresponding $i-$th block of rows of $M_{\Phi(G)}$ has the following
structure:
\begin{eqnarray*}
\left[ \begin{array}{cccc|cccccc}
  b_{i1} & b_{i2} & \cdots & b_{i(k-g_c)} &b_{i(k-g_c+1)} &\cdots&
b_{i(k-g_i+1)} & \cdots & b_{i(k-2)} & b_{i(k-1)}\\
  b_{i2} & b_{i3} & \cdots &b_{i(k-g_c+1)} & b_{i(k-g_c+2)} &
\cdots&b_{i(k-g_i+2)} &
\cdots & b_{i(k-1)} & *\\
  \vdots & \vdots & \vdots & \cdots & \vdots & \vdots & \vdots & \vdots&
\vdots&\vdots \\
  b_{ig_i} & b_{i(g_i+1)} & \cdots &b_{i(k-g_c+g_i-1)} & b_{i(k-g_c+g_i)}
&\cdots &
b_{i(k-1)} & \cdots & * & *\\
 \end{array}\right]\in \F_q^{g_i\times (k-1)}
\end{eqnarray*}
The $*$s shown in the blocks above represent entries from $\F_q$ which depend on
the matrix $A_{k-1}$. Since $g_c\geq g_i$ $\forall i$, the $*$s appear only in
the last $g_c-1$ columns of $M_{\Phi(G)}$ (shown as the trailing submatrix after the vertical line). As $\Phi(G)-$extension of $W$ has
rank $g-1$, therefore $M_{\Phi(G)}$ has rank $g-1$. 

Similarly, corresponding to the matrix state $M^*$ of $W^*$, we have the
following matrix state of the $G-$extension of $W^*$.
\begin{eqnarray*}
 M_{G}^* &=& [\textcolor{green}{x_1^*;x_1^*A_{k};\ldots;x_1^*A_{k}^{g_1-1}};
\textcolor{red}{x_2^*;x_2^*A_{k};\ldots;x_2^*A_{k}^{g_2-1}}; \ldots;
 \textcolor{blue}{x_{c-1}^*;x_{c-1}^*A_{k};\ldots;x_{c-1}^*A_{k}^{g_{c-1}-1}}
;\\&&
e_{k}^{k};e_{k}^{k}A_{k};\ldots;e_{k}^{k}A_{k}^{g_c-1}
;\textcolor {cyan} {
x_{c+1}^*;x_{c+1}^*A_{k};\ldots;x_{c+1}^*A_{k}^{g_{c+1}-1}};
\ldots;\textcolor{magenta}{
x_m^*;x_m^*A_{k};\ldots;x_m^*A_{k}^{g_m-1}}]\\
&& \textrm{where } A_{k} \textrm{ is the companion matrix of the polynomial
}p_{k}(s).
\end{eqnarray*}

 For $i \neq c$, the $i-$th block of $M_G^*$ is
$[x_i^*;x_i^*A_{k};\ldots;x_i^*A_{k}^{g_i-1}]$ (recall that $x_i^* =
(x_i,d_i)$), where $A_{k}$ is the companion matrix of the polynomial
$p_{k}(s)$. This block has the following structure
\begin{eqnarray*}
\left[ \begin{array}{cccc|cccccc}
  b_{i1} & b_{i2} & \cdots & b_{i(k-g_c)} & \cdots & b_{i(k-g_i+1)} &
 b_{i(k-g_i +2)} &\cdots  & b_{i(k-1)} & \textcolor{red}{d_i}\\
  b_{i2} & b_{i3} & \cdots & b_{i(k-g_c+1)} & \cdots& b_{i(k-g_i+2)} &
 b_{i(k-g_i+3)} &\cdots  & \textcolor{red}{ d_i} & *\\
  \vdots & \vdots & \vdots & \cdots & \vdots & \vdots & \vdots & \vdots&
\vdots & {\vdots}\\
  b_{ig_i} & b_{i(g_i+1)} & \cdots & b_{i(k-g_c+g_i-1)} &\cdots &
b_{i(k-1)} & \textcolor{red}{d_i}  & \cdots &  * & {*}\\
 \end{array}\right]\in \F_q^{g_i\times (k)}
\end{eqnarray*}

 The $c-$th block of $M_G^*$ is
$[e_{k}^{k};e_{k}^{k}A_{k};\ldots;e_{k}^{k}A_{k}^{g_c-1}]$.
This block has the following structure.
 \begin{eqnarray*}
\left[ \begin{array}{cccc|cccc}
  0 & 0 & \cdots & 0 & 0 & \cdots & 0 & 1\\
  0 & 0 & \cdots & 0 & 0 & \cdots & 1 & *\\
  \vdots & \vdots & \vdots & \cdots & \vdots & \vdots & \vdots & \vdots\\
  0 & 0 & \cdots &  0& 1 & \cdots & * & *\\
 \end{array}\right]\in \F_q^{(g_c) \times (k)}
\end{eqnarray*}

Let $M_G$ be the submatrix of $M_G^*$ got by removing its last column and the
first row of its $c-$th block. Observe that $rank(M_G) = rank(M_G^*)-1$. By the
structure of the $c-$th block of
$M_G$ one can clearly see that this submatrix $M_G$ can be modified to $M_{\Phi(G)}$
using elementary row operations. Hence this submatrix $M_G$ has rank $g-1$. This
implies that $M_G^*$ has rank $g$. Therefore, $W^*$ does have a $G-$extension with dimension $g$.

Note that each of the $d_i$s can be chosen in
$q$ ways. Each such choice yields a different matrix $M^*$ and hence a
different multisequence $W^*$. As a
result {\bf for every multisequence $W$ with minimal polynomial $p_{k-1}(s)$,
the above process gives us $q^{m-1}$
multisequences $W^*$ with minimal polynomial $p_{k}(s)$}. 
Therefore, 
\begin{equation}
\label{eqn1}
 N(G,k) \geq q^{m-1}N(\Phi(G),k-1)
\end{equation}

Conversely, consider a multisequence $U^*$ in $\F_q^m$ with primitive minimal
polynomial $p_k(s)$ whose $G-$extension has rank $g$. Consider its matrix state
$M_1^* \in \F_q^{m \times k}$ whose $c-$th row is $e_k^k$. Now $M_1^*$ can be
reduced to a matrix  $M_1 \in \F_q^{m \times (k-1)}$ as follows:
\begin{enumerate}
 \item For $i \neq c$ remove the last entry of the $i-$th row.
 \item Let the $c-$th row of $M_1$ be $e_{k-1}^{k-1}$
\end{enumerate}
 
Let $M_1$ generate a multisequence $U$ having primitive minimal polynomial $p_{k-1}(s)$. Using similar arguments as those used earlier in the proof, one can prove that
the $\Phi(G)-$extension of $U$ has dimension $g-1$. 
 Note that the matrix $M_1$ is independent  of  the last entries of the rows of
$M_1^*$. Hence, there are $q^{m-1}$ matrices (including $M_1^*$),
 with $c-$th row $e_{k}^{k}$, which have the same first $k-1$ columns as
 $M_1$.  By the above process each one of these matrices 
  gives the same matrix $M_1$ (and hence the same multisequence $U$).
 Besides if we start with a matrix with $c-$th row $e_{k}^{k}$ which differs
from $M_1$ in any entry corresponding to the first $k-1$ columns, it results
in a different $M_1$ (and hence a different multisequence $U$). Therefore,
 \begin{eqnarray}
 \label{eqn2}
 N(\Phi(G),{k-1}) &\geq& \frac{ N(G,{k})}{q^{m-1}}  \\
 \Rightarrow q^{m-1}N(\Phi(G),{k-1}) &\geq& N(G,{k}) \nonumber
 \end{eqnarray}
 Thus, from equations \eqref{eqn1} and \eqref{eqn2} we can conclude that
 \begin{eqnarray*}
  N(G,{k}) = q^{m-1}N(\Phi(G),{k-1})
 \end{eqnarray*}

 \end{proof}

Using this result, Theorem \ref{maintheorem} can by proved in the following
manner:
\begin{proof}[Proof of Theorem \ref{maintheorem}]
 For each $j$, such that $n-r+m \leq j \leq n$, let $p_j(s)$ be a given
 primitive polynomial of degree $j$. For every point $G= (g_1,g_2,\ldots,g_m)$
on the $R-$road, let $g = \sum_{i=1}^mg_i$. As we have seen in the proof of
Theorem \ref{starLemma}, starting from a multisequence in $\F_q^m$ with
dimension
$m$ (i.e., its ${\bf 1}-$extension has maximum dimension), having minimal
polynomial $p_{n-r+m}(s)$, we can recursively generate multisequences in
$\F_q^m$, with minimal polynomial $p_{n-r+g}(s)$, whose $G-$extensions have
maximum dimension, for every $G$ on the $R-$road.

By Theorem \ref{starLemma}, for any two consecutive points, $\Phi(G)$ and $G =
(g_1,g_2,\ldots,g_m)$ in the path from ${\bf 1}=(1,1,\ldots,1)$ to $R$,
$N(G,{n-r+g}) = q^{m-1}N(\Phi(G),{n-r +g-1})$ where $g = \sum_{i=1}^mg_i$.
The path from ${\bf 1}$ to $R$ has $r-m$ such steps. Therefore,
\begin{equation}
 N(R,n) = (q^{m-1})^{r-m}N({\bf 1},{n-r+m})
\end{equation}
 However, $N({\bf 1},{n-r+m})$ is the number of multisequences in
$\F_q^m$ of dimension $m$, with a given primitive minimal polynomial $p_{n-r+m}(s)$ of degree
$n-r+m$.
Therefore, by Corollary \ref{maincor}, $ N({\bf 1},{n-r+m}) =
(q^{n-r+m}-q)(q^{n-r+m}-q^2)\ldots(q^{n-r+m}-q^{m-1})$. Hence,
\begin{eqnarray*}
 N(R,{n})&=&
(q^{m-1})^{r-m}(q^{n-r+m}-q)(q^{n-r+m}-q^2)\ldots(q^{n-r+m}-q^{m-1})\\
                           &=&  (q^n-q^{r-m+1})(q^n - q^{r-m+2})\ldots(q^n -
q^{r-1})
\end{eqnarray*}
Hence proved.
\end{proof}

\begin{remark}
 $N(R,{n})$ does not depend on the integers $(r_1,r_2,\ldots,r_m)$ but just
their sum.
\end{remark}
One can therefore ask the following question.

\begin{problem}
 Given any $r\geq m$, how many multisequences in $\F_q^r$ having dimension $r$
are $R-$extensions of multisequences in $\F_q^m$ for some $R
=(r_1,r_2,\ldots,r_m)\in \mathbb{Z}_+^m$ where $\sum r_i = r$. 
\end{problem}

This problem is answered in the following lemma.

\begin{lemma}
The number of multisequences in $\F_q^r$ which are $R-$extensions of
multisequences in $\F_q^m$ is given by,
\begin{equation}
 N_r = \binom {r-1}{r-m} (q^n-q^{r-m+1})(q^n - q^{r-m+2})\ldots(q^n -
q^{r-1})
\end{equation}
\end{lemma}

\begin{proof}
 For any $r \in \mathbb{Z}_+$, define the following subset $\mathfrak{R}_r $ of
$\mathbb{Z}_+^m$. 
\begin{eqnarray*}
\mathfrak{R}_r := \{(r_1,r_2,\ldots,r_m) \in \mathbb{Z}_+^m ~|~ \sum_{i=1}^mr_i
= r\}
\end{eqnarray*}

Therefore, 
\begin{eqnarray*}
 N_r = |\mathfrak{R}_r| \times (q^n-q^{r-m+1})(q^n - q^{r-m+2})\ldots(q^n -
q^{r-1})
\end{eqnarray*}
Corresponding to each element of $\mathfrak{R}_r $, say $(r_1,r_2,\ldots,r_m)$,
we can define a monomial, $x_1^{r_1}x_2^{r_2}\ldots x_m^{r_m}$. Therefore,
calculating  $|\mathfrak{R}_r|$ is equivalent to finding the number of monomials
of degree $r$ where each variable 
is raised to a nonzero index. However, every such monomial can be written as
$x_1x_2\ldots x_m \Upsilon(x_1,x_2,\ldots, x_m)$, where
$\Upsilon(x_1,x_2,\ldots, x_m)$ is a monomial of degree $r-m$. Consequently, the
cardinality of 
$\mathfrak{R}_r$ is equal to the number of monomials of degree $r-m$. This
number is equal to $\binom{(r-m)+m-1}{r-m} = \binom {r-1}{r-m}$. As a result, 
\begin{eqnarray*}
 N_r = \binom {r-1}{r-m}  (q^n-q^{r-m+1})(q^n - q^{r-m+2})\ldots(q^n -
q^{r-1})
\end{eqnarray*}
\end{proof}

Given $R=(r_1,r_2,\ldots,r_m) \in \mathbb{Z}_+^m$, let $r = \sum_{i=1}^mr_i$.
Let $\{p_j(s)\}_{j = n-r+m}^n$ be a series of primitive polynomials where the
index $j$ denotes the degree of the respective polynomial. Let $A_j$s be
their corresponding companion matrices. Let $\Phi(G)$ and $G$ be consecutive
points on the $R-$road. Further, let $c$ be the position of the active
coordinate of $\Phi(G)$. Consider a multisequence $U$ in $\F_q^m$ with a minimal
polynomial $p_{n-r+g-1}(s)$, whose $\Phi(G)-$extension has maximum dimension.
Note that $U$ can be uniquely determined by any of its matrix states. Let $M_U$
be its matrix state whose $c-$th row is $e_{n-r+g-1}^{n-r+g-1}$. The proof of
Theorem \ref{starLemma} gives us a procedure to go from $M_U$, to the matrix
state $M_U^*$ of a multisequence $U^*$ with a primitive minimal polynomial
$p_{n-r+g}(s)$ of degree $n-r+g$, whose $G-$extension has maximum dimension.
Thus, using this procedure one can generate a sequence of matrices
$\{M_j\}_{j=n-r+m}^m$ starting with a matrix $M_{n-r+m} \in \F_q^{m \times
(n-r+m)}$ having full row rank, and culminating in a matrix  $M_n \in \F_q^{m
\times (n)}$. Each matrix $M_j$ in the above sequence uniquely corresponds to a
point $G$ on the $R-$road and can be seen as a matrix state of a multisequence
with minimal polynomial $p_j(s)$ whose corresponding $G-$extension has maximum
dimension. The following is an algorithm to generate this sequence:


\begin{algorithm}
\label{algo}
(The variable $M$ is used to store the respective matrix state at every step of
the algorithm. The current point in the path from ${\bf 1}$ to $R$ is stored in
the variable $G = (g_1,g_2,\ldots,g_m)$. The variable $c$ stores the position of
the active coordinate of $G$. The variable $g$ stores the summation of the
values of the coordinates of $G$. ) \\
{\bf Initialization:}
\begin{itemize}
 \item Initialize $G$ to ${\bf 1}$.  
 \item Initialize the value of $g$ to $m$.
 \item Initialize $M$ to any matrix in $\F_q^{m \times (n-r+m)}$ that has full
row rank.  
\end{itemize}
{\bf Main Loop:}
 \begin{itemize}
    \item While $g < r$ do the following
    \begin{itemize}
    \item Find the position of the active coordinate of $G$ and store it
in $c$. 
    \item Find a polynomial $f(s)$ such
that $M(c,:)f(A_{n-r+g}) = e_{n-r+g}^{n-r+g}$.
                       \item $M = M f(A_{n-r+g})$. (This gives us the matrix
state whose $c-$th row is $e_{n-r+g}^{n-r+g}$).
                       \item  For all $i \neq c$, append the $i-$th row of $M$
with any  $d_i \in \F_q$ to get the row vector $(M(i,:),d_i)$. Change the $c-$th
row of $M$ to $e_{n-r+g+1}^{n-r+g+1}$.
                       \item Increment $g$ and $g_c$ by $1$. 
       \end{itemize}
   \end{itemize}
\end{algorithm}
The most complex part of the above algorithm is to find the polynomial $f(s)$.
This can be done using the following subroutine.
\begin{subroutine}
\begin{itemize}
 \item Construct the matrix $\mathscr{M} =
[M(c,:);M(c,:)A_{n-r+g};\ldots;M(c,:)A_{n-r+g}^{n-r+g-1}]$.
 \item Solve the set of linear equations 
\begin{equation}
\label{polycalc}
a\mathscr{M} = e_{n-r+g}^{n-r+g} \,\,\,\, \textrm{ for }\,\, a \in \F_q^{n-r+g}.
\end{equation}
 \item If $a = (a_0,a_1,\ldots,a_{n-r+g-1})$ is the solution to the above set of
equations, $a_0M(c,:) + a_1M(c,:)A_{n-r+g} + \cdots +
a_{n-r+g-1}M(c,:)A_{n-r+g}^{n-r+g-1} = e_{n-r+g}^{n-r+g}$. Therefore $f(s) = a_0
+ a_1s + \cdots + a_{n-r+g-1}s^{n-r+g-1}$.
\end{itemize}
\end{subroutine}

Let $c_1$ and $c_2$ be the active coordinates of ${\bf 1}$
and $\Phi(R)$ respectively. Algorithm \ref{algo} can be thought of as a map from the space of
matrices in $\F_q^{m \times (n-r+m)}$ which have full row rank and whose
$c_1-$th rows are $e_{n-r+m}^{n-r+m}$, to the space of matrices in $ \F_q^{m
\times n}$ which have full row rank and whose $c_2-$th rows are $e_n^n$.  There are precisely $ (q^{n-r+m}-q)(q^{n-r+m}-q^2)\ldots(q^{n-r+m}-q^{m-1})$
matrices in $\F_q^{m \times (n-r+m)}$ whose $c_1-$th row
is $e_{n-r+m}^{n-r+m}$. During each iteration of the while loop one can chose $d_i$s 
in $q^{m-1}$ ways. Therefore, corresponding to each choice of matrix $M_{n-r+m} \in \F_q^{m \times (n-r+m)}$ there are $q^{(m-1)(r-m)}$
possible candidates for $M_n \in \F_q^{m \times n}$. No two distinct choices for the matrix $M_{n-r+m}$ can give the
same $M_n$. 
Therefore, we have $
(q^n-q^{r-m+1})(q^n - q^{r-m+2})\ldots(q^n -q^{r-1})$ possible matrices which
can occur as an output to Algorithm \ref{algo}. The
number of full row rank matrices in $\F_q^{m \times n} $, whose $c_2-$th row is
$e_n^n$, is however $(q^n - q)(q^{n}-q^2)\ldots(q^{n}-q^{m-1})$.
Out of these matrices, precisely those matrices that occur as matrix
states of multisequences whose $R-$extensions have full rank are the ones that can be obtained from the above algorithm.

We now determine the computational complexity of Algorithm
\ref{algo}. We begin by evaluating the computational complexity of calculating
$Mf(A_{n-r+g})$: 

\begin{itemize}
 \item For any $i$, $M(i,:)A_{n-r+g}^j = (M(i,:)A_{n-r+g}^{j-1})A_{n-r+g}$.
Therefore, knowing $(M(i,:)A_{n-r+g}^{j-1})$, $M(i,:)A_{n-r+g}^j$ can
be calculated in $O(n^2)$ steps. Consequently, the matrix $\mathscr{M}$ can be
generated in $O(n^3)$ operations.
\item Using $LU-$decomposition, solving the set of linear equations
\eqref{polycalc} takes $O(n^3)$ operations.
\item For any $i$, such that $(1\leq i \leq n-r+g)$, the $i-$th
row of $Mf(A_{n-r+g})$ is given by $a_0M(i,:) + a_1M(i,:)A_{n-r+g} + \cdots +
a_{n-r+g-1}M(i,:)A_{n-r+g}^{n-r+g-1} $. As we have already seen each element in
the above summation can be generated in $O(n^2)$ steps. Each row of
$Mf(A_{n-r+g})$ can thus be calculated in $O(n^3)$ steps. Therefore, 
$Mf(A_{n-r+g})$ can be calculated in $O(n^4)$ operations.
\end{itemize}
As we have already seen, the active coordinate of $G$ can be found in $O(m)$
steps. Also, appending  $m-1$ rows of $M$ takes $O(m)$ operations. Further,
incrementing $g$ and $g_c$ has complexity $O(1)$. Therefore each iteration of the
while loop takes $O(m+n^4)$ time.  Since the number of steps from ${\bf 1}$ to
$R$ is $(r-m)$, the while loop runs for a maximum of $r-m$ iterations.
 As a result, {\bf the computational complexity of algorithm \ref{algo} is
$O((r-m)(m+n^4))$}.
 Therefore for a fixed $m$ and $r$ this computational complexity is $O(n^4)$. This is therefore a polynomial time algorithm.    

 We now proceed to see an application of the above developed theory.

\section{Word Based Linear Feedback Shift Registers}
The theory developed in the preceding sections finds an application in word
based Linear Feedback Shift Register (LFSR) design. We begin our discussion by
giving a brief introduction to Linear Feedback Shift
Registers (LFSR)s. 

LFSRs are electronic circuits that implement LRRs. These are widely used
in the field of pseudo-random number generation and coding theory. LFSRs
consist
of delay elements, feedback elements and adders. 
 {\small
 \begin{figure*}[h]
  \label{spiderman}
 \psfrag{a0}{\tiny $a_0$}
 \psfrag{a1}{\tiny $a_1$}
 \psfrag{a2}{\tiny $a_2$}
 \psfrag{ak-3}{\tiny $a_{n-3}$}
 \psfrag{ak-2}{\tiny $a_{n-2}$}
 \psfrag{ak-1}{\tiny $a_{n-1}$}
 \psfrag{D0}{\tiny $D_0$}
 \psfrag{D1}{\tiny $D_1$}
 \psfrag{D2}{\tiny $D_2$}
 \psfrag{Dk-3}{\tiny $D_{n-3}$}
 \psfrag{Dk-2}{\tiny $D_{n-2}$}
 \psfrag{Dk-1}{\tiny $D_{n-1}$}
 \psfrag{op}{\tiny Output}
       \centering
       \includegraphics[scale=0.6]{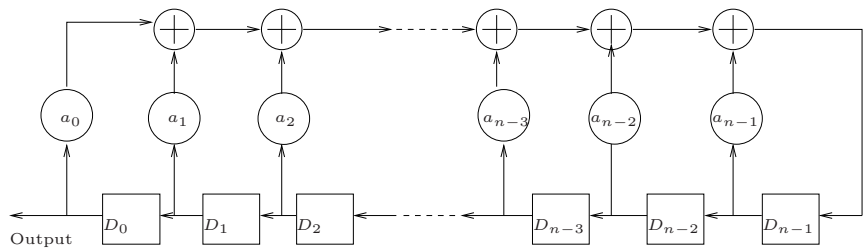}
 \begin{center}
  \caption{Linear Feedback Shift Register}
 \end{center}
  \end{figure*}
 }
For example, the LFSR corresponding to the LRR $S({k+n}) = a_{n-1}S(k+n-1)+
 a_{n-2}S({k+n-2}) + \cdots + a_0S(k)$ is as shown in Figure 1. LFSRs with
primitive characteristic polynomials are of particular interest since they
generate
sequences with desirable randomness properties like  $2$-level
autocorrelation property and span-$n$ property (all nonzero subsequences of
length $n$ occur once in every period)\cite{Golomb}.
An LFSR can be seen as a state machine where the states are the outputs of the
delay blocks. Its state transition matrix is the companion matrix of the
characteristic polynomial of the LRR, (i.e., matrix $A$ in equation \eqref{A}).

Conventional LFSRs use bitwise operations and hence are incapable of 
efficiently utilizing the parallelism provided by word based processors. In the
1994 conference on fast software encryption, a challenge was set
forth to design LFSR's which exploit the parallelism offered by the word
oriented operations of modern processors \cite{Preneel}.  A
special case of this scheme was implemented by Tsaban and Vishne
in their  paper \cite{Tsabman}. Here, they
introduced a family of efficient word oriented LFSRs with multiple input
multiple output delay blocks. 
The design of Tsaban and Vishne was further generalized in \cite{zeng} wherein
the structure shown in figure \ref{f4} was proposed to implement the
mathematical scheme proposed in \cite{Neider2}. 
\begin{figure*}[h]
\psfrag{a0}{\tiny $B_0$}
\psfrag{a1}{\tiny $B_1$}
\psfrag{ak-2}{\tiny $B_{b-2}$}
\psfrag{ak-1}{\tiny $B_{b-1}$}
\psfrag{D0}{\tiny $D_0$}
\psfrag{D1}{\tiny $D_1$}
\psfrag{Dk-2}{\tiny $D_{b-2}$}
\psfrag{Dk-1}{\tiny $D_{b-1}$ }
\psfrag{mbits}{\tiny $m-$bits }
\psfrag{op}{\tiny Output}
     
\begin{center}
\includegraphics[scale=0.6]{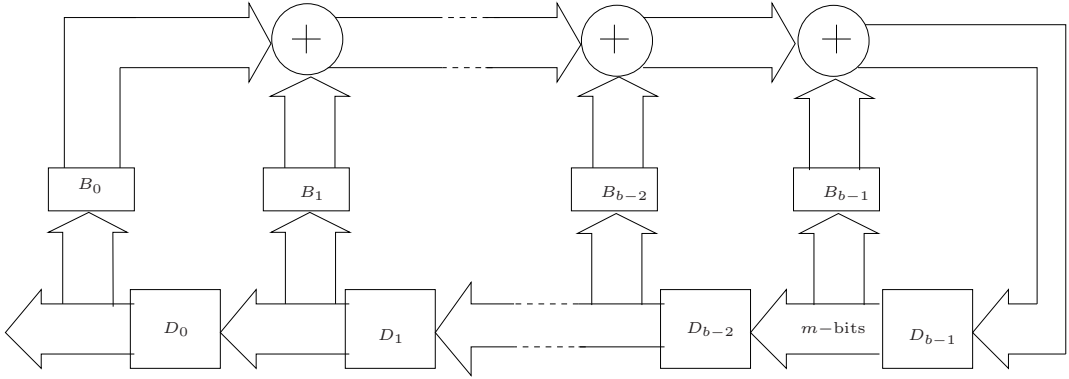}
 \caption{Linear Feedback Shift Register with m-Input m-Output Delay Blocks}
\label{f4}
\end{center}      
\end{figure*}

Consider the LFSR shown in Figure \ref{f4} 
Let $W(k)$ be the output
of the LFSR at the
$k$-th time instant. Due to the structure of the LFSR, the following algebraic
relation is satisfied by the vectors generated by it. 
\begin{equation}
\label{mimolrr}
 W({k+b}) = B_0W(k) + B_1W({k+1}) + \cdots + B_{b-1}W({k+b-1})
\end{equation}
where $B_i \in \F_q^{m \times m}$. Therefore, for all $i$,
\begin{eqnarray}
\label{recurrence}
\left[ \begin{matrix} W({k+1})\\W({k+2})\\ \vdots
\\W({k+b}) \end{matrix}\right] =  A_{mb}\left[ \begin{matrix} W(k)\\W({k+1})\\
\vdots \\W({k+b-1})
 \end{matrix}\right] 
\end{eqnarray}
where,
\begin{eqnarray*}
A_{mb} = \left[
\begin{matrix}
 0 & I & 0 & \ldots & 0 \\
 0 & 0 & I & \ldots & 0   \\
\vdots & \vdots & \vdots & \ddots & \vdots  \\
0 & 0 & 0 & \ldots & I  \\
B_0&B_1&B_2& \ldots & B_{b-1}
 \end{matrix}
\right] \in \F_q^{mb \times mb}
\end{eqnarray*}
We henceforth call the structure of the matrix $A_{mb}$ as the $m-$companion
structure. The matrix $A_{mb}$ is called the transition matrix of the LFSR and
it uniquely characterizes the LFSR.  The characteristic
polynomial
of an LFSR is the characteristic polynomial of the respective transition matrix.
As in the scalar case,  LFSRs with primitive characteristic polynomials are
of special interest. Interestingly, in this design, different combinations of
feedback matrices can be used to get the same characteristic polynomial. This
gives rise to the following question:
 \begin{problem}
\label{mainproblem}
 Given integers $n,m$ and a primitive polynomial $p(s)$, of degree $n = mb$, how
many different LFSR realizations, using m-input m-output delay elements, have
$p(s)$ as their characteristic polynomial?
\end{problem}
From the discussion above, it is clear that this
number is equal to the number of m-companion matrices that have the given
primitive polynomial as their characteristic polynomial. 
Therefore, problem \ref{mainproblem} can be restated as follows.
\begin{problem}
  Given integers $n,m$ and a primitive polynomial $p(s)$, of degree $n = mb$,
how many $m-$companion matrices in $\F_q^{n \times n}$ have $p(s)$ as their
characteristic polynomial?
\end{problem}

This question was addressed by us in \cite{HP} where the solution was found for
the cases $m=1$, $m=2$ and $m=n$. We will now demonstrate how the theory
developed earlier in the paper can be used to solve this problem for the
general case and in addition give an algorithm for generating such
configurations.

Consider an LFSR with $m-$input $m-$output delay blocks with
primitive characteristic polynomial $p(s)$ of degree $n=mb$. Let $A_{mb}$ be
the transition matrix of the LFSR.
The output of such an LFSR can be seen as a
multisequence $W= \{W(k)\}_{k \in \mathbb{Z}}$ in $\F_q^m$. 
Let $W_1,W_2,\ldots,W_m$ be the component sequences of $W$. Each component
sequence of $W$ satisfies the LRR corresponding to $p(s)$. Therefore, the
minimal polynomial of $W$ divides $p(s)$. Since $p(s)$ is a primitive
polynomial, the minimal polynomial of $W$ is $p(s)$. The matrix state of $W$
at time instant $k$ is $M(k) =
[W(k),W(k+1),\ldots,W(k+n-1)]$. Let $x_i(k)$ be the state vector of the
component sequence $W_i$ at time instant $k$. Therefore $M(k)$ can also be
written as $[x_1(k);x_2(k);\ldots;x_m(k)]$. Let $R =
\underbrace{(b,b,\ldots,b)}_{m-\textrm{times}}$ and let $W_R$ be the
$R-$extension of $W$. Therefore, the matrix state of $W_R$ at time instant $k$
is $M_R(k) = [x_1(k);\sigma x_1(k);\ldots;\sigma^{b-1}x_1(k);x_2(k);\sigma
x_2(k);\ldots;\sigma^{b-1}x_2(k),\ldots,x_m(k);\sigma
x_1(k);\ldots;\sigma^{b-1}x_m(k)]$. By permuting the rows of $M_R(k)$ we get the
matrix $M_1(k) = [x_1(k),x_2(k),\ldots,x_m(k),\sigma x_1(k),\sigma x_2(k),$ $
\sigma x_m(k), \ldots, \sigma^{b-1}x_1(k),\sigma^{b-1}x_2(k),$
$\ldots,\sigma^{b-1}x_m(k)]$ which can also be written as follows:
\begin{eqnarray*}
 M_1(k) &=& \left[\begin{matrix}
 W(k) & W({k+1}) & \cdots & W({k+mb-1})\\
 W({k+1}) & W({k+2}) & \cdots & W({k+mb})\\
 \vdots & \vdots & \cdots & \vdots\\
 W({k+b-1}) &  W({k+b}) & \cdots & W({k+b+mb-2})
\end{matrix}\right]\\ &=& 
\left[ \begin{matrix} W({k})\\W({k+1})\\ \vdots
\\W({k+b-1}) \end{matrix},A_{mb}\left(\begin{matrix}
W({k})\\W({k+1})\\ \vdots
\\W({k+b-1}) \end{matrix}\right),\cdots, A_{mb}^{n-1}\left(\begin{matrix}
W({k})\\W({k+1})\\
\vdots
\\W({k+b-1}) \end{matrix}\right)\right]
\end{eqnarray*}
Since $A_{mb}$ has a primitive characteristic polynomial, for any non zero
vector $v \in \F_q^n$, the vectors $v,A_{mb}v,\ldots,A_{mb}^{n-1}v$ are
linearly independent. Therefore, $M_1(k)$ has rank $n$. As a consequence,
$M_R(k)$ has rank $n$ i.e., $W_R$ has dimension $n$. We therefore have the
following lemma:

\begin{lemma}
Let $R = \underbrace{(b,b,\ldots,b)}_{m-\textrm{times}}$. Given an LFSR with
$m-$input $m-$output delay blocks having a primitive characteristic polynomial
$p(s)$ of degree $n$, the $R-$extension of any non zero multisequence generated
by it has an dimension $n$ (i.e., maximum dimension).
\end{lemma}

Further, since $A_{mb}$ has primitive characteristic polynomial, repeated
action of $A_{mb}$ on any nonzero vector $v\in \F_q^n$ will generate all
nonzero vectors in $\F_q^n$. Therefore, by Equation \eqref{recurrence}, for any
initial nonzero state of the LFSR all possible nonzero states of the LFSR will
be covered. As a result, all multisequences generated by the LFSR are just
shifted versions of each other. In other words, each LFSR with a primitive
characteristic polynomial has a unique multisequence associated with it.

Conversely, consider a multisequence $W^*= \{W^*(k)\}_{k \in \mathbb{Z}}$ with
primitive minimal polynomial $p(s)$ of degree $n$, whose $R-$extension has
dimension $n$. Let $A$ be the companion matrix of $p(s)$ and $M_W^*(k)$ be the 
matrix state of $W^*$ at instant $k$. One can construct the following full rank
matrices $M^*(k)$ by permuting the rows of the matrix states of the
$R-$extension of $W^*$.

{\small
 \begin{eqnarray}
\label{consmat}
\left[\begin{matrix} M_W^*(k)\\ M_W^*(k)A\\ \vdots \\
M_W^*(k)A^{b-1} \end{matrix}\right]  &=& \left[\begin{matrix}
 W^*(k) & W^*({k+1}) & \cdots & W^*({k+mb-1})\\
 W^*({k+1}) & W^*({k+2}) & \cdots & W^*({k+mb})\\
 \vdots & \vdots & \cdots & \vdots\\
 W^*({k+b-1}) &  W^*({k+b}) & \cdots & W^*({k+b+mb-2})
\end{matrix}\right] = M^*(k)
 \end{eqnarray}}

Clearly, for any  $k \in \mathbb{Z}$, $M^*(k+1) = M^*(k)A = M^*(k-1)A^2 = \cdots
= M^*(0)A^{k+1}$. Therefore,
$M^*(k+1)M(k)^{-1} = M^*(0)A^{k+1}A^{-k}M^*(0)^{-1} = M^*(0)AM^*(0)^{-1}$. Thus
 $M^*(k+1)M(k)^{-1}$ is independent of $k$ and is a constant matrix for a given
multisequence $W^*$. Let this matrix be denoted by $A_{mb}^*$. Therefore, for
all $k$, 
\begin{equation}
 \label{invrec}
M^*(k+1) = A_{mb}^*M^*(k).
\end{equation}
Further, given any matrix state $M_W^*(k)$ of $W^*$, $A_{mb}^*$ can be
constructed as follows: 
\begin{equation}
 \label{invrec}
A_{mb}^* = M^*(k)AM^*(k)^{-1}.
\end{equation}
where $M^*(k)$ is got from Equation \eqref{consmat}.

It can be easily verified that the matrix $A_{mb}^*$ satisfying equation
\eqref{invrec} has an $m-$canonical structure. Let $A_{mb}^*$ be as follows: 
\begin{eqnarray*}
A_{mb}^* = \left[
\begin{matrix}
 0 & I & 0 & \ldots & 0 \\
 0 & 0 & I & \ldots & 0   \\
\vdots & \vdots & \vdots & \ddots & \vdots  \\
0 & 0 & 0 & \ldots & I  \\
B_0^*&B_1^*&B_2^*& \ldots & B_{b-1}^*
 \end{matrix}
\right] \in \F_q^{mb \times mb}
\end{eqnarray*}
Therefore, for any $k$, $W^*(k+b) = B_0^*W^*(k) + B_1^*W^*(k+1) + \cdots +
B_{b-1}^*W^*(k+b-1)$. Thus we have an LRR (and hence an LFSR) generating the  
the multisequence $W$.

Thus, we have demonstrated a one to one correspondence between LFSRs with
$m-$input $m-$output delay blocks having a given primitive minimal polynomial
$p(s)$ of degree $n=mb$ and multisequences in $\F_q^m$ with minimal polynomial
$p(s)$ whose $R-$extensions have dimension $n$. Therefore, by Theorem
\ref{maintheorem} we have the following:

\begin{theorem}
Number of LFSRs with $m-$input $m-$output delay blocks whose transition
matrices have a given primitive polynomial $p(s)$ of degree $n=mb$, as their
characteristic polynomial is
$(q^{n}-q^{n-1})(q^{n}-q^{n-2})\ldots(q^{n}-q^{n-m+1})$
\end{theorem}

Thus, for ${R}=\underbrace{(b,b,\ldots, b)}_{m \textrm{ times}}$, every
multisequence $W$ with a primitive minimal polynomial $p(s)$, such that $W_{R}$
has dimension $mb$, is generated by a unique LFSR whose transition matrix
has characteristic polynomial $p(s)$. Besides, given a matrix state of
 $W$ one can uniquely determine the transition
matrix $A_{mb}$ of the LFSR by Equation \eqref{invrec}.

Therefore, the problem of finding LFSRs generating multisequences with a given
primitive polynomial reduces to a special case of problem \ref{mainProblem}
where $n = r=
mb$ and ${R}=\underbrace{(b,b,\ldots b)}_{m \textrm{ times}}$. Hence algorithm
\ref{algo} can be used to obtain desired LFSR configurations as demonstrated in
the following example.

\subsection{Example}
We demonstrate Algorithm \ref{algo} by generating a $3-$companion matrix over $\F_2$ with primitive characteristic polynomial $p_6(s) = s^6 +s +1$. We therefore generate a multisequence in $\F_2^3$ whose
$(2,2,2)-$extension has maximum dimension i.e., $6$. Note that here $R = (2,2,2)$. 
Consider
\begin{eqnarray*}
 p_5(s) &=& s^5 + s^2 + 1\\
 p_4(s) &=& s^4 + s +1\\
 p_3(s) &=& s^3 + s+1
\end{eqnarray*} which are primitive polynomials over $\F_2$ of corresponding degrees. Let $A_i$s be the companion matrices of the respective $p_i(s)$s, for $3\leq i\leq 6$. 
We start with a multisequence in $\F_q^3$ with minimal polynomial $p_3(s)$ and the following matrix state
\begin{eqnarray*}
 M_3 = \left[\begin{matrix}
  1 & 1 & 1\\
  0 & 1 & 1\\
  0 & 0 & 1 
 \end{matrix}
\right]
\end{eqnarray*}
We initialize $G={\bf
1} = (1,1,1)$
\begin{enumerate}
 \item[Iteration 1:] $G = (1,1,1)$. Therefore, the active coordinate of $G$
is the $3$rd coordinate. (The $3$rd row of $M$ is already $e_3^3$ and hence
$M_3$ is the desired matrix state).  
Let us append the first and second rows of $M_3$ with $1$ and $0$ respectively and
change the third row to $e_4^4$. We therefore get the matrix.
\begin{eqnarray*}
  \left[\begin{matrix}
  1 & 1 & 1 & 1\\
  0 & 1 & 1 & 0\\
  0 & 0 & 0 & 1 
 \end{matrix}
\right]
\end{eqnarray*}
  This is the state matrix of a multisequence $W_4$ with characteristic
polynomial $p_4(s)$. Increment the active coordinate of $G$ to get $G =
(1,1,2)$. It can be verified that $W_4$ is a multisequence whose $(1,1,2)-$extension has dimension $4$.

 \item[Iteration 2:] $G = (1,1,2)$. Therefore, the active coordinate of $G$ is
the $2$nd coordinate. The matrix state of $W_4$ with second row being $e_4^4$
is 
\begin{eqnarray*}
 M_4 =  \left[\begin{matrix}
  1 & 1 & 0 & 1\\
  0 & 0 & 0 & 1\\
  0 & 1 & 1 & 1 
 \end{matrix}
\right]
\end{eqnarray*}
Suppose we append the first and third rows of $M_4$ with $0$ and $1$ respectively and
change the second row to $e_5^5$. This gives the following matrix:
 \begin{eqnarray*}
   \left[\begin{matrix}
  1 & 1 & 0 & 1 & 0\\
  0 & 0 & 0 & 0 & 1\\
  0 & 1 & 1 & 0 & 1 
 \end{matrix}
\right]
\end{eqnarray*}
 This is the state matrix of a multisequence $W_5$ with characteristic
polynomial $p_5(s)$. Increment
the active coordinate of $G$ to get $G =
(1,2,2)$. Note that $W_5$ is a multisequence whose $(1,2,2)-$extension has dimension $5$.

 \item[Iteration 3:] $G = (1,2,2)$. Therefore, the active coordinate of $G$ is
the $1$st coordinate. The matrix state of $W_5$ with first row being $e_5^5$
is 
\begin{eqnarray*}
 M_5 =  \left[\begin{matrix}
  0 & 0 & 0 & 0& 1\\
  0 & 1 & 0 & 1 & 1\\
  0 & 0 & 0 & 1 & 1 
 \end{matrix}
\right]
\end{eqnarray*}
Suppose we append the second and third rows of $M_5$ with $0$ and $1$ respectively and
change the second row to $e_6^6$. This gives the following matrix:
 \begin{eqnarray*}
 M_W =  \left[\begin{matrix}
  0 & 0 & 0 & 0 & 0 & 1\\
  0 & 1 & 0 & 1 & 1 & 0\\
  0 & 0 & 0 & 1 & 1 & 1
 \end{matrix}
\right]
\end{eqnarray*}
 This is the state matrix of a multisequence $W$ with characteristic
polynomial $p_6(s)$. Increment the active coordinate of $G$ to get $G =
(2,2,2)$. Now $W$ is a multisequence whose $(2,2,2)-$extension has maximum
dimension, i.e., $6$.
\end{enumerate}
%
%

Using the matrix $M_W$ we can construct the following matrix 
$M^*$: 
\begin{eqnarray*}
 M^* = \left[\begin{matrix}
                  M_W\\
                  M_WA_{6}
                 \end{matrix}\right] = \left[
\begin{matrix}
 0 & 0 & 0 & 0 & 0 & 1\\
 0 & 1 & 0 & 1 & 1 & 0\\
 0 & 0 & 0 & 1 & 1 & 1\\
 0 & 0 & 0 & 0 & 1 & 0\\
 1 & 0 & 1 & 1 & 0 & 1\\
 0 & 0 & 1 & 1 & 1 & 0\\
\end{matrix}\right]
\end{eqnarray*}
The $3-$companion matrix $A_{33}$ can now be obtained as follows:
\begin{eqnarray*}
 A_{33} = M^*A_{6}(M^*)^{-1} = 
\left[\begin{array}{ccc|ccc}
       0 & 0 & 0 & 1 & 0 & 0\\
       0 & 0 & 0 & 0 & 1 & 0\\
       0 & 0 & 0 & 0 & 0 & 1\\
       \hline
       1 & 0 & 1 & 1 & 0 & 0\\
       1 & 1 & 0 & 1 & 0 & 1\\ 
       1 & 1 & 1 & 1 & 0 & 1
      \end{array}\right]
\end{eqnarray*}
This corresponds to an LFSR whose output multisequence will satisfy the
following linear recurring recurring relation: 
\begin{eqnarray*}
W(k+2) = \left[\begin{matrix} 
              1 & 0 & 1\\
              1 & 1 & 0\\
              1 & 1 & 1
             \end{matrix}\right]W(k)+
\left[\begin{matrix} 
              1 & 0 & 0\\
              1 & 0 & 1\\
              1 & 0 & 1
             \end{matrix}\right]W(k+1)
\end{eqnarray*}

 \section{Counting the number of non-singular block Hankel matrices}
Hankel matrices are specially structured matrices which frequently
appear in the fields of signal processing \cite{Hasan}, image
processing and control theory. In this section
we derive a formula for the number of non-singular Hankel matrices of a given
size, over a given finite field $\F_q$, by using the theory developed in the
preceding sections. 

A Hankel matrix is a matrix which is constant along the anti-diagonals. For
example:
\begin{eqnarray}
\label{Hankel}
 H = \left[\begin{matrix}
   a_1 & a_{2} &  \ldots & a_{n-1}& a_n\\
   a_2 & a_3   &  \ldots & a_n & a_{n+1}\\
  \vdots & \vdots & \vdots & \vdots & \vdots \\
   a_{n-1} & a_{n} &  \ldots &  a_{2n-3}& a_{2n-2}\\
   a_n & a_{n+1} &  \ldots &  a_{2n-2}& a_{2n-1}
 \end{matrix}\right]
\end{eqnarray}
 It can be easily seen that the space of $n \times n$ Hankel matrices is a
$2n-1$ dimensional space and each Hankel matrix can be uniquely determined by
the corresponding vector $a_H = (a_1,a_2,\ldots,a_{2n-1}) \in \F_q^{2n-1}$. 
Since the number of Hankel matrices over a finite field $\F_q$ is finite, one
may pose the following question.

\begin{problem}
Given $n$, find the number of Hankel matrices in $\F_q^{n \times n}$ that
have full rank.
\end{problem}

We solve this problem by proving a bijection between the set of full rank
Hankel matrices in $\F_q^{n \times n}$ and the set of multisequences in $\F_q^2$
with a given primitive minimal polynomial $p(s)$ whose $R-$extensions have
maximum dimension, for $R= (n-1,n)$.

\begin{theorem}
Let $p(s)$ be a primitive polynomial of degree $2n-1$. Let $R = (n-1,n)$.
Consider a Hankel matrix $H\in \F_q^{n \times n}$ corresponding to the vector $
a_H = (a_1,a_2,\ldots,a_{2n-1}) \in \F_q^{2n-1}$. The matrix $H$ has full rank
if and only if the matrix $M = [e_{2n-1}^{2n-1}; a_H]$ is a matrix state of a
multisequence $W$ with minimal polynomial $p(s)$, whose $R-$extension has
maximum dimension. 
\end{theorem}
\begin{proof}
Consider the multisequence $W$ with primitive polynomial $p(s)$ and matrix state
$M$ . Therefore the corresponding matrix state $M_R$ of the $R-$extension of $W$
is as follows:

\begin{eqnarray*}
 M_R= \left[
 \begin{array}{ccccc|cccc}
0 & 0 &\ldots&  0 & 0 & 0 &\ldots& 0 & 1\\
0 & 0 &\ldots&  0 & 0 & 0 &\ldots &1 & *\\
\vdots &\vdots  & \vdots & \vdots & \vdots & \vdots & \vdots & \vdots
& \vdots \\
0 & 0 &\ldots&  0 & 0 & 1 &\ldots& * & *\\
\hline
a_1 & a_2 & \ldots & a_{n-1} & a_{n}  & a_{n+1} & \ldots & a_{2n-2}
& a_{2n-1}\\ 
a_2 & a_3 & \ldots & a_{n} & a_{n+1}  & a_{n+2} & \ldots & a_{2n-1}
& *\\ 
\vdots &\vdots & \vdots &  \vdots & \vdots & \vdots & \vdots & \vdots
& \vdots \\ 
a_{n-1} & a_{n} &\ldots &  a_{2n-3} & a_{2n-2} & a_{2n-1} & * & * & *
\\
a_n & a_{n+1} &\ldots &  a_{2n-2} & a_{2n-1} & * & * & * & * 
\end{array}\right]
\end{eqnarray*}
Clearly, the submatrix $M_R(1:n-1,n+1 : 2n-1)$ (the top right
submatrix) has full rank. Therefore, $M_R$ has full rank if and only if the
submatrix $M_R(n:2n-1,1 : n)$ (the bottom left submatrix) has full rank.
However, $M_R(n:2n-1,1 : n)$ is the Hankel matrix $H$. Hence, the matrix $H$
has full rank if and only if the matrix $M_R$ has full rank. In other words, the
matrix $H$ has full rank if and only if the $R-$extension  of the multisequence
$W$ has maximum dimension. Hence proved.
\end{proof}

Every multisequence $W$ in $\F_q^2$, with primitive minimal polynomial $p(s)$
is uniquely characterized by any of its matrix states. Therefore, the number of
full rank Henkel matrices in $\F_q^{n \times n}$
is equal to the number of multisequences in $\F_q^2$ whose $R-$extensions have
maximum dimension. Hence, by Theorem \ref{maintheorem} we have the following
theorem

\begin{theorem}
The number of Hankel matrices in $\F_q^{n \times n}$ having full rank is
$(q^{2n-1}- q^{2n-2})$. 
\end{theorem}

\section{Conclusions}
In this paper we have introduced the concept of matrix states. Using matrix
states, we have defined the dimension of a multisequence and calculated the
number of multisequences with a given dimension. The concept of $R-$extensions
has then been introduced. We have calculated the number of multisequences
whose
$R-$extensions have maximum dimension. Further we give an algorithm to
generate such multisequences.  We have then demonstrated an application of the
theory developed for $R-$extensions. For any given $m$, we have derived a
formula for the number of LFSR configurations, with $m$ input $m$ output
delay blocks, that generate multisequences with a given primitive minimal
polynomial. Further, we have demonstrated the use of the algorithm developed for
$R-$extensions, for the generation of such LFSR configurations. Finally, using
the theory developed, we have derived a formula for the number of Hankel
matrices in $\F_q^{n \times n}$ that have full rank.
       
 \nopagebreak
\bibliographystyle{ieeetr}
\bibliography{LSFR}

\begin{thebibliography}{10}

\bibitem{Golomb}
S.~W. Golomb, {\em Shift Register Sequences}.
\newblock Cambridge University Press, 1967.

\bibitem{Schneier}
B.~Schneier, {\em Applied Cryptography: protocols,algorithms and source code in
  {C}}.
\newblock John Wiley and Sons Inc, New York, 1996.

\bibitem{peterson}
W.~W. Peterson, {\em Error Correcting Codes}.
\newblock John Wiley and Sons Inc, New York, 1961.

\bibitem{pickholtz}
R.~Pickholtz, D.L.Schilling, and L.~B. Milstein, ``Theory of spread-spectrum
  communications - a tutorial,'' {\em IEEE Transactions on Communications},
  vol.~30, no.~5, pp.~855--884, 1982.

\bibitem{Daykin}
D.~E. Daykin, ``On linear sequences over a finite field,'' {\em The American
  Mathematical Monthly}, 1962.

\bibitem{Mullen}
W.~S. Chou and G.~L. Mullen, ``Generating linear span over finite fields,''
  {\em ACTA Mathematica}, pp.~183--191, 1992.

\bibitem{Yucas}
R.Fitzgerald and J.Yucas, ``On generating linear span over gf(p),'' {\em Congr.
  Numer.}, vol.~69, pp.~55--60, 1989.

\bibitem{Ecuyer}
L'Ecuyer, ``Random number for simulation,'' {\em Comm.ACM}, vol.~33, no.~10,
  pp.~85--97, 1990.

\bibitem{Neider1}
R.~Neiderreiter, ``The multiple recursive matrix method for pseudorandom vector
  generation,'' {\em Finite Fields and Application}, vol.~3, no.~30, 1995.

\bibitem{Neider2}
R.~Neiderreiter, ``Pseudorandom vector generation by multiple recursive matrix
  method,'' {\em Mathematics of Computation}, vol.~64, no.~209, pp.~279--294,
  1995.

\bibitem{lidl}
R.~Lidl and H.~Neiderrieter, {\em Introduction to Finite Fields and their
  Applications}.
\newblock Cambridge University Press, 1986.

\bibitem{Preneel}
B.~Preneel, ``Introduction,'' in {\em Proc. Fast Software Encryption 1994
  Workshop (Lecture Notes in Computer Science)}, vol.~1008, pp.~1--5,
  Springer-Verlag, 1995.

\bibitem{Tsabman}
B.~Tsaban and U.~Vishne, ``Efficient linear feedback shift registers with
  maximal period,'' {\em Finite Fields and Their Applications}, vol.~8, no.~2,
  pp.~256--267, 2002.

\bibitem{zeng}
G.~Zeng, W.~Han, and K.~He, ``High efficiency feedback shift register:
  $\sigma$-{LFSR},'' {\em IACR Eprint archive}, 2007.

\bibitem{HP}
S.~Krishnaswamy and H.~K. Pillai, ``On the number of linear feedback shift
  registers with a special structure,'' {\em IEEE Transactions on Information
  Theory}, vol.~58, pp.~1783--1790, March 2012.

\bibitem{Hasan}
M.~A. Hassan and A.~A. Hassan, ``Hankel matrices of finite rank with
  applications to signal processing and polynomials,'' {\em Journal of
  Mathematical Analysis and Applications}, vol.~208, pp.~218--242, 1997.

\end{thebibliography}

\end{document}